\documentclass[11pt]{article}
\usepackage{multirow}
\usepackage{amssymb}
\usepackage{amsmath}
\usepackage{amsthm}
\usepackage{thm-restate}
\usepackage[dvipsnames]{xcolor}

\usepackage[pdftex, plainpages = false, pdfpagelabels, 
                 bookmarks=false,
                 bookmarksopen = true,
                 bookmarksnumbered = true,
                 breaklinks = true,
                 linktocpage,
                 pagebackref,
                 colorlinks = true,  
                 linkcolor = blue,
                 urlcolor  = blue,
                 citecolor = red,
                 anchorcolor = green,
                 hyperindex = true,
                 hypertexnames = false,
                 hyperfigures
                 ]{hyperref}
\usepackage{complexity} 
\usepackage{enumitem}
\setlist{nosep}

\usepackage{tikz}

\usetikzlibrary{arrows,shapes,positioning,shadows,trees}

\usepackage{vmargin}
\setmarginsrb{1in}{1in}{1in}{1in}{0pt}{0pt}{0pt}{6mm}

\newcommand{\OO}{\mathcal{O}}

\newtheorem{theorem}{Theorem}
\newtheorem{lemma}{Lemma}

\newtheorem{claim}{Claim}[section]

\newtheorem{proposition}{Proposition}[section]
\newtheorem{reduction rule}{Reduction Rule}[section]
\newtheorem{marking-scheme}{Marking Scheme}[section]

\newenvironment{claimproof}[1]{\par\noindent\textit{Proof of Claim.}\space#1}{\hfill $\diamond $\medskip}

\usepackage{algorithmic}
\usepackage{algorithm}
\usepackage[utf8]{inputenc}
\usepackage[LGR,T1]{fontenc} 
\usepackage[T1]{fontenc}
\usepackage{amsmath}

\newcommand{\blob}{\textnormal{\texttt{blob}}}

\newcommand{\maxpartial}{\textnormal{\textsc{Maximum Partial List $H$-Coloring}}}

\newcommand{\lst}{\textnormal{\texttt{list}}}

\newcommand{\wt}{\textnormal{\texttt{wt}}}
\newcommand{\hm}{\textnormal{\texttt{hom}}}

\newcommand{\size}{\textnormal{\texttt{size}}}

\title{Maximum Partial List $H$-Coloring on $P_5$-free graphs in polynomial time}
\author{Daniel Lokshtanov\thanks{University of California, Santa Barbara, USA \href{mailto:daniello@ucsb.edu}{daniello@ucsb.edu}.}
\and
Pawe\l{} Rz\k{a}\.{z}ewski\thanks{
Warsaw University of Technology, Warsaw, Poland and 
Institute of Informatics, University of Warsaw, Warsaw, Poland \href{mailto:pawel.rzazewski@pw.edu.pl}{pawel.rzazewski@pw.edu.pl}.}
\and Saket Saurabh\thanks{Institute of Mathematical Sciences, Chennai, India and University of Bergen, Norway \href{mailto:saket@imsc.res.in}{saket@imsc.res.in}.}
\and Roohani Sharma\thanks{Discrete Mathematics Group, Institute for Basic Science (IBS), Daejeon, South Korea \href{mailto:roohani@ibs.re.kr}{roohani@ibs.re.kr}. Supported by the Young Scientist Fellowship (YSF) of the Institute for Basic Science (IBS-R029-Y8).}
\and 
Meirav Zehavi\thanks{Ben-Gurion University of the Negev, Beer-sheva, Israel \href{mailto:meiravze@bgu.ac.il}{meiravze@bgu.ac.il}.}
}
\date{}

\begin{document}

\maketitle

\begin{abstract}
    In this article we show that \maxpartial\ is polynomial-time solvable on $P_5$-free graphs for every fixed graph $H$. In particular, this implies that {\sc Maximum $k$-Colorable Subgraph} is polynomial-time solvable on $P_5$-free graphs. This answers an open question from Agrawal, Lima, Lokshtanov, Rz\k{a}\.{z}ewski, Saurabh \& Sharma [SODA 2024]. This also improves the $n^{\OO(\omega(G))}$-time algorithm for {\sc Maximum Partial $H$-Coloring}, where $\omega(G)$ is the size of the largest clique in $G$,
    by Chudnovsky, King, Pilipczuk,  Rz\k{a}\.{z}ewski \& Spirkl [SIDMA 2021],
    to polynomial-time algorithm (independent of the maximum clique size of the graph).
\end{abstract}

\section{Introduction}\label{sec:intro}
In the \maxpartial\ problem, $H$ is a fixed graph without loops
and, 
given a graph $G$ together with two functions, the weight function $\wt: V(G) \to \mathbb{Q}_{\geq 0}$ and the list function $\lst: V(G) \to 2^{V(H)}$,
the goal is to find an induced subgraph $G^*$ of $G$ such that $\wt(V(G^*))$ is maximum and there is a homomorphism from $G^*$ to $H$ respecting the $\lst$ function. That is, there exists a function $\hm: V(G^*) \to V(H)$ such that for each $uv \in E(G)$, $\hm(u)\hm(v) \in E(H)$. Such a subgraph $G^*$ is called a {\em solution} of maximum weight with respect to $\wt$.

When $H$ is a clique on $k$ vertices, and $\lst(v) = V(H)$ for each $v \in V(G)$, \maxpartial\ is the {\sc Maximum $k$-Colorable Subgraph} problem. Further, for $k =2$, this is the {\sc Maximum $2$-Colorable Subgraph} problem or equivalently the dual of the {\sc Odd Cycle Transversal} problem.

In a recent work~\cite{DBLP:journals/talg/0001LLR0S25},
{\sc Odd Cycle Transversal} was shown to be polynomial-time solvable on $P_5$-free graphs, thereby answering the last open case in understanding the P versus NP-hard dichotomy for {\sc Odd Cycle Transversal} on $H$-free graphs, when $H$ is connected. 
Here $P_5$ is a path on $5$ vertices.
It was posed as an open question in~\cite{DBLP:conf/soda/AgrawalLLSS24},
whether {\sc Maximum $k$-Colorable Subgraph} problem is polynomial-time solvable on $P_5$-free graphs.
In this work, we resolve this question in positive by giving a polynomial-time algorithm for a more general problem \maxpartial.
Our result also considerably improves the algorithm for
{\sc Maximum Partial $H$-Coloring} by Chudnovsky, King, Pilipczuk,  Rz\k{a}\.{z}ewski \& Spirkl~\cite{DBLP:journals/siamdm/ChudnovskyKPRS21} who gave an algorithm running in time $n^{\OO(\omega(G))}$, where $\omega(G)$ is the maximum clique number of the graph $G$.

\begin{theorem}\label{thm:main}
    \maxpartial\ can be solved in $n^{\OO(k^4)}$ time  on $P_5$-free graphs, where $k=|V(H)|$.
\end{theorem}

We would like to remark that in a recent and independent work, Henderson, Smith-Roberge, Spirkl \& Whitman~\cite{henderson2024maximumkcolourableinducedsubgraphs} show that the {\sc Maximum $k$-Colorable Subgraph} problem is polynomial-time solvable on $(P_5 + rK_1)$-free graphs. While their algorithm works for a more general graph class, our algorithm works for a more general problem.

To prove Theorem~\ref{thm:main}, 
we mimic the general approach of~\cite{DBLP:journals/talg/0001LLR0S25} of designing a polynomial-time algorithm for {\sc Odd Cycle Transversal} on $P_5$-free graphs, for \maxpartial.
The main part of the approach of~\cite{DBLP:journals/talg/0001LLR0S25} is showing a polynomial-time algorithm that finds a polynomial-sized family of connected vertex sets of the input graph such that there exists a maximum weight solution which is a disjoint union of some connected sets in this family, and any disjoint union of connected sets in this family is a solution.
We prove a similar result for \maxpartial, by following the general approach in~\cite{DBLP:journals/talg/0001LLR0S25} and adapting its steps to the more general setting.
The major challenge here is to deal with a scenario where one is promised that a maximum weight solution exists which is also connected. With this extra promise, the setup of~\cite{DBLP:journals/talg/0001LLR0S25} immediately finds a maximum weight solution for the {\sc Odd Cycle Transversal} problem: 
 one can find two disjoint sets here and argue that there is a solution which is the union of two arbitrary independent sets from each of these sets. 
 Such independent sets of a $P_5$-free graph can be found in polynomial time using the algorithm of~\cite{DBLP:conf/soda/LokshantovVV14}.

But for the case of \maxpartial, the situation, even in the case when a connected solution is guaranteed, is far more complicated.
For this situation, in Section~\ref{sec:conn}, we design a polynomial-time procedure that reduces the problem to solving it on a few instances where the maximum size of the list of any vertex (which is upper bounded by $|V(H)|$) of these instances is strictly smaller than that of the input list. 
This allows to then call the main algorithm recursively, at most $|V(H)|$-many times.

The full approach of finding a polynomial-sized family of connected vertex sets, in polynomial-time, with a guarantee that a maximum weight solution for \maxpartial\ can be obtained by taking the disjoint union of some of its connected sets, and that a disjoint union of any subset of its connected sets guarantees a solution, 
is given in Section~\ref{sec:nonconn}.
Finally, the full proof of Theorem~\ref{thm:main} appears in Section~\ref{sec:proof-of-theorem}.

\section{Preliminaries}\label{sec:prelims}
Throughout the paper, $n$ denotes the number of vertices in the input graph $G$. 
For a set $X$ and positive integers $i,j_1,j_2$, we denote by $2^X$ the collection of all subsets of $X$, by ${X \choose i}$ the collection of all $i$-sized subsets of $X$, by ${X \choose \leq i}$ all subsets of $X$ of size at most $i$, and by ${X \choose j_1 \leq i \leq j_2}$ the collection of all $i$-sized subsets of $X$ where $j_1 \leq i \leq j_2$.
For any function $f : A \to B$ and any $C \subseteq A$, by $f_{|C}$ we represent the function $f$ whose domain is restricted to $C$.

For any graph $G$, for any $v \in V(G)$, by $N_G(v)$ we denote the {\em open neighbourhood} of $v$ in $G$, that is $N_G(v)$ is the set of vertices in $G$ which are adjacent to $v$ in $G$.
By $N_G[v]$, we denote the set of {\em closed neighbourhood} of $v$ in $G$, that is $N_G[v]:= N_G(v) \cup \{v\}$.
For any $X \subseteq V(G)$, the set $N_G(X):= (\bigcup_{v \in X} N_G(v)) \setminus X$ is the open neighbourhood of $X$ in $G$ and the set $N_G[X]:= (\bigcup_{v \in X} N_G(v)) \cup X$ is the closed neighbourhood of $X$ in $G$.
Whenever the graph $G$ is clear from the context, we drop the subscript $G$ in the above notation.

Given the $\lst: V(G) \to 2^{V(H)}$ function, by the {\em size of $\lst$}, denoted by $\size(\lst)$, we mean $\max_{v \in V(G)} |\lst(v)|$.

For any graph $G$ and set $X \subseteq V(G)$, the set $X$ is called a {\em module} in $G$ if for each $u,v \in X$, $N_G(u) \setminus X = N_G(v) \setminus X$.
For $X,Y \subseteq V(G)$, $E(X,Y)$ denotes the set of edges with one endpoint in $X$ and the other in $Y$. 
For a connected component $C$ of $G$, we sometimes abuse notation and write $C$ for its vertex set too when the context is clear.
A {\em dominating set} of $X$ in $G$ is a vertex subset $Y \subseteq X$ such that for each $u \in X$, either $u \in Y$ or $N_G(u) \cap Y \neq \emptyset$.
A dominating set of a graph $G$ is simply a dominating set of $V(G)$.
For any positive integer $t$, $P_t$ is a path on $t$ vertices.

We will use the following two facts about $P_5$ graphs. Proposition~\ref{prop:dom} says that connected $P_5$-free graphs have a dominating set which either induces a path on three vertices or a clique. Proposition~\ref{prop:is} says that {\sc Maximum Weight Independent Set} problem on $P_5$-free graphs can be solved in polynomial time.

\begin{proposition}[Theorem~$8$, \cite{bacso1990dominating}]\label{prop:dom}
Every connected $P_5$-free graph $G$ has a dominating set $D$ such that $G[D]$ is either a $P_3$ or a clique.
\end{proposition}

\begin{proposition}[\cite{DBLP:conf/soda/LokshantovVV14}, Independent Set on $P_5$-free]\label{prop:is}
    Given an $n$-vertex graph $G$ with $m$ edges and a weight function $\wt:V(G) \to \mathbb{Q}_{\geq 0}$, there is a $\OO(n^{12}m)$-time algorithm that outputs a set $I \subseteq V(G)$ such that $I$ is an independent set in $G$ and $\wt(I)=\sum_{u \in I}\wt(u)$ is maximum.
\end{proposition}

\section{Inductive step when there exists a connected solution}\label{sec:conn}
Recall that the input instance for the \maxpartial\ problem is $(G,\wt,\lst)$. 
Let $k=|V(H)|$.
For convenience, we will refer to the solution as a vertex set $C$ of $G$ instead of an induced subgraph $G^*$. In this case, we mean to say that the solution is the graph $G$ induced on the vertex set $C$.

Roughly speaking, in this section we show that, assuming our input instance admits a maximum weight solution which is connected, the problem reduces to solving instances with strictly smaller $\lst$ size. Since the maximum size of $\lst$ is at most $|V(H)|$, this allows for a recursive algorithm of depth $|V(H)|$, which is presented in the next section.

More concretely, in this section, we assume that $(G,\wt,\lst)$ admits a maximum weight solution $C$ that is connected and show that, in time $n^{\OO(k^3)}$, one can output an $n^{\OO(k^3)}$-sized family $\mathcal{F}$ containing elements of the form $(D, (G_1,\wt_1,\lst_1), \ldots, (G_p,\wt_p,\lst_p))$, where $D \subseteq V(G)$, $p \leq k$, for each $i \in \{1, \ldots, p\}$, $G_i$ is an induced subgraph of $G$, $\wt_1: V(G_1) \to \mathbb{Q}_{\geq 0}$ and $\lst_i : V(G_i) \to 2^{V(H)}$. 
The guarantee of $\mathcal{F}$ is that if $(G,\wt,\lst)$ admits a maximum weight solution that is connected, then there exists an element $(D, (G_1,\wt_1,\lst_1), \ldots, (G_p,\wt_p,\lst_p))$ in the family $\mathcal{F}$ such that there exists a maximum weight solution $S$ of $(G,\wt,\lst)$ which is a union of $D$ and an arbitrary maximum weight solution of $(G_i,\wt_i,\lst_i)$, for each $i \in \{1, \ldots, p\}$.
Furthermore, for each element $(D, (G_1,\wt_1,\lst_1), \ldots, (G_p,\wt_p,\lst_p))$ of $\mathcal{F}$, for each $i \in \{1, \ldots, p\}$,
$\size(\lst_i) < \size(\lst)$.

\begin{algorithm}[t] 
\caption{List-size reduction} 
\label{alg:size-reduction} 
\begin{algorithmic}[1] 
    \REQUIRE An undirected graph $G'$,
    a weight function $\wt : V(G') \to \mathbb{Q}_{\geq 0}$ and 
    a list function $\lst: V(G') \to 2^{V(H)}$
    \ENSURE Either $C \subseteq V(G')$ such that $C$ is a maximum weight solution of $(G',\wt,\lst)$ or family $\mathcal{F}$ satisfying the properties of Lemma~\ref{lem:conn}. 
    \FORALL{$D \subseteq {V(G') \choose {i \leq \max\{k,3\}}}$}\label{line:for-one-reduce}
            \STATE Set $G=G'$ and let $D=\{d_1, \ldots, d_{|D|}\}$.\label{line:set-one-reduce}
            \STATE Consider the following partition of $V(G) = (D \uplus X_1 \uplus \ldots, X_{|D|} \uplus R)$, where $X_1= N_G(d_1) \setminus D$, for each $i \in \{2, \ldots, |D|\}$, $X_i = N_G(d_i) \setminus (X_1 \cup \ldots \cup X_{i-1} \cup D)$, and $R= V(G) \setminus (X_1 \cup \ldots \cup X_{|D|} \cup D)$.\label{line:partition-reduce}
            \STATE Update $G=G-R$.\label{line:delete-R}
            \FORALL{$1 \leq i <j \leq |D|$, $r \in \{1, \ldots, k\}$ and $\widetilde{X}_{i,j}^r \subseteq {X_i \choose {\leq 2}}$ such that $\widetilde{X}_{i,j}^r$ is an independent set in $G$}\label{line:for-two-reduce}
                    \WHILE{$\exists u \in \widetilde{X}_{i,j}^r, \exists v \in X_j$ such that $uv \in E(G)$, $r \in \lst(u)$, $r' \in \lst(v)$ where $r' =r$ or $rr' \not \in E(H)$}\label{line:while-one-reduce-start}
                            \STATE Update $\lst(v) = \lst(v) \setminus \{r'\}$.
                    \ENDWHILE \label{line:while-one-reduce-end}
                    \WHILE{$\exists u \in X_i, \exists v \in X_j$ such that $uv \in E(G)$, $r \in \lst(u)$, $r' \in \lst(v)$ where $r' =r$ or $rr' \not \in E(H)$} \label{line:while-two-reduce-start}
                            \STATE Update $\lst(u) = \lst(u) \setminus \{r\}$.
                    \ENDWHILE \label{line:while-two-reduce-end}
                    \FORALL{$h: D \to \{1, \ldots, k\}$ such that if $h(d_p) =q$, then $q \in \lst(d_p)$}\label{line:for-three-reduce}
                            \STATE Update $\lst(d_i)=\{h(d_i)\}$.
                    \ENDFOR\label{line:for-three-end-reduce}
                    \FORALL{$d_p \in D$, $v \in N_G(d_p)$}
                            \STATE Update $\lst(v) = \lst(v) \setminus \{\{h(d_p)\} \cup \{r' : h(d_p)r' \not \in E(H)\} \}$.
                    \ENDFOR\label{line:for-four-end-reduce}
                    \IF{$\size(\lst) \leq 1$}\label{line:if-size-one-reduce}
                            \FORALL{$r \in \{1, \ldots, k\}$} 
                                    \STATE Set $V_r =\{ v \in V(G) : \lst(v) = \{r\}\}$.
                                    \STATE Let $I_r$ be a maximum weight independent set in $G[V_r]$ with respect to the weight function $\wt_{|V_r}$, computed using the polynomial-time algorithm of Proposition~\ref{prop:is}.\label{line:algo-is-reduce}
                                    \STATE Update $\mathcal{C} = \mathcal{C} \cup \{\bigcup_{r=1}^k I_r\}$.
                            \ENDFOR \label{line:if-size-one-reduce-end}
                    \ELSE \label{line:if-size-two-reduce}
                            \STATE Update $\mathcal{F} = \mathcal{F} \cup \{(D, (G[X_1], \wt_{|X_1}, \lst_{|X_1}), \ldots, (G[X_{|D|}], \wt_{|X_{|D|}}, \lst_{|X_{|D|}})\}$. \label{line:if-size-two-reduce-end}
                    \ENDIF
            \ENDFOR
    \ENDFOR   
    \IF{$|\size(\lst)| \leq 1$}
        \RETURN $C' \in \mathcal{C}$ such that $\wt(C')$ is maximum.\label{line:output-solution}
    \ELSE
        \RETURN $\mathcal{F}$. \label{line:output-family}
    \ENDIF
\end{algorithmic}
\end{algorithm}

\begin{lemma}\label{lem:conn}
    Given a graph $H$ and an instance $(G',\wt,\lst)$ of \maxpartial\ such that $G$ is $P_5$-free,
     Algorithm~\ref{alg:size-reduction} is an
     $n^{\OO(k^3)}$-time algorithm, where $k=|V(H)|$, that does the following.
     \begin{itemize}
         \item Either correctly outputs a maximum weight solution of $(G', \wt, \lst)$, or
         \item outputs an $n^{\OO(k^3)}$-sized family $\mathcal{F}$, each of whose elements are of the form $(D, (G'_1,\wt_1,\lst_1), \allowbreak  \ldots, \allowbreak (G'_p,\wt_p,\lst_p))$, 
    where $D \subseteq V(G')$, $p \leq k$, 
    for each $i \in \{1, \ldots, p\}$, 
    $G'_i$ is an induced subgraph of $G$, 
    $\wt_i: V(G'_i) \to \mathbb{Q}_{\geq 0}$, 
    $\lst_i : V(G'_i) \to 2^{V(H)}$ 
    and $\size(\lst_i) < \size(\lst)$ such that the following holds:
    if $(G',\wt,\lst)$ has a connected solution of maximum weight, then there exists some maximum weight solution $S$ of $(G',\wt,\lst)$,
    and some element $(D, (G'_1,\wt_1,\lst_1), \ldots, \allowbreak (G'_p,\wt_p,\lst_p))$ of $\mathcal{F}$,
    such that $S$ is the union of $D$ with an arbitrary maximum weight solution of $(G'_i,\wt_i,\lst_i)$, for each $i \in \{1, \ldots, p\}$.
    \end{itemize}
     
\end{lemma}
\begin{proof}
    Note that the statement of the lemma assumes an existence of a connected solution, but the solution returned in the first point may not be connected.  Without loss of generality let $V(H)=\{1, \ldots,k\}$.

    We first argue about the running time of Algorithm~\ref{alg:size-reduction}.
    Observe that except for the for-loops at Lines~\ref{line:for-one-reduce}, \ref{line:for-two-reduce} and~\ref{line:for-three-reduce},
    all the other steps run in $n^{\OO(1)}$ time (including the algorithm of Proposition~\ref{prop:is} at Line~\ref{line:algo-is-reduce}). The for-loop at Line~\ref{line:for-one-reduce} runs for $n^{\OO(k)}$ times, at Line~\ref{line:for-two-reduce} runs for $n^{\OO(k^3)}$ times and at Line~\ref{line:for-three-reduce} runs for $k^{\OO(k)}$ time. 
    Since $k \leq n$ (without loss of generality),
    the running time of Algorithm~\ref{alg:size-reduction} is $n^{\OO(k^3)}$.

    In the following, we argue about the correctness of Algorithm~\ref{alg:size-reduction}.
    Note that Algorithm~\ref{alg:size-reduction} either reports a set $C' \subseteq V(G')$ in Line~\ref{line:output-solution} or a family $\mathcal{F}$ at Line~\ref{line:output-family}.

    \vspace{10pt}

    \noindent{{\bf Lines~\ref{line:for-one-reduce}-\ref{line:set-one-reduce}}}
    Let $C \subseteq V(G)$ be a connected solution of maximum weight for the instance $(G', \allowbreak \wt, \allowbreak  \lst)$.
    Since $G'[C]$ has a (list) homomorphism to $H$,
    the maximum size of a clique in $G'[C]$ is at most $|V(H)|$ ($=k$).
    Further since $G'[C]$ is $P_5$-free, from Proposition~\ref{prop:dom}
    there exists a dominating set $D=\{d_1, \ldots, d_{|D|}\}$ of $C$ of size at most $\max\{k,3\}$. Consider the run of the for-loop at Line~\ref{line:for-one-reduce} for this choice of $D$. Set $G=G'$.

    \vspace{10pt}

    \noindent{{\bf Lines~\ref{line:partition-reduce}-\ref{line:delete-R}}}
    Consider the partition $(D,X_1, \ldots, X_{|D|},R)$ of $V(G)$ from Line~\ref{line:partition-reduce}.
    The only place where we will use the guarantee that there exists a maximum weight solution which is connected, is to claim that $R$ does not intersect with the solution, for the correct guess of the dominating set $D$. 
    Indeed, for any $v \in R$, the vertex $v$ does not belong to $C$,
     because $D$ dominates $C$ and $N[D] \subseteq D \cup \bigcup_{i=1}^{|D|}X_i$.
    Therefore one can safely delete all vertices in $R$. Henceforth, assume that $R = \emptyset$. In fact, $(D,X_1, \ldots, X_{|D|})$ is a partition of $V(G)$.

    \vspace{10pt}

    \noindent{{\bf Line~\ref{line:for-two-reduce}}}
    \begin{claim}\label{claim:ind-matching-size-3}
        For any $i,j \in \{1, \ldots, |D|\}$, $i <j$, 
        if $I$ is an independent set in $X_i$ and $P$ are the neighbors of $I$ in $X_j$, then there exists $I' \subseteq I$ of size at most $2$ that dominates $P$.
    \end{claim}
     \begin{claimproof}
        For contradiction, let $I'$ be a minimum-sized subset of $I$ that dominates $P$ and $|I'| \geq 3$. Then there exists $a,b,c \in I'$ and $x,y,z \in P$ such that $ax,by,cz \in E(G)$ but $ay,az,bx,bz,cx,cy \not \in E(G)$. 

        We consider two exhaustive cases.

        Case 1: $G[\{x,y,z\}]$ is edgeless. In this case $x,a,d_i,b,y$ is an induced $P_5$ in $G$ (note that $d_i$ has no neighbors in $X_j$ by the construction of $X_j$ and because $j >i$).

        Case 2: $G[\{x,y,z\}]$ has an edge. Without loss of generality, let $xy \in E(G)$. In this case $y,x,a,d_i,c$ is an induced $P_5$ in $G$.
    \end{claimproof}

Let $\hm : C \to V(H)$ be a homomorphism from $G[C]$ to $H$ that respects the list function \lst.
Recall that $V(H)=\{1, \ldots, k\}$.
For each $r \in \{1, \ldots, k\}$,
let $C_r:= \hm^{-1}(r)$. Then $G[C_r]$ is edgeless, that is $C_r$ is an independent set in $G$.
For each $i \in \{1, \ldots, |D|\}$ and $r \in \{1, \ldots, k\}$,
let $X_i^r:= X_i \cap C_r$. Note that $X_i^r$ is an independent set.
Further, for each $j \in \{i+1, \ldots, |D|\}$,
from Claim~\ref{claim:ind-matching-size-3}, there exists at most $2$ vertices in $X_i^r$ that dominates all neighbors of $X_i^r$ in $X_j$. 
Let the set of these two vertices be $\widetilde{X}_{i,j}^r$. 
Then, as just argued, from Claim~\ref{claim:ind-matching-size-3}, $(N(X_i^r) \cap X_j) \subseteq (N(\widetilde{X}_{i,j}^r) \cap X_j)$.

For each $i \in \{1, \ldots, |D|\}$, $j \in \{i+1, \ldots, |D|\}$ and $r \in \{1, \ldots, k\}$, 
consider the run of the for-loop at Line~\ref{line:for-two-reduce} for these choices of the sets 
$\widetilde{X}_{i,j}^r$ which need to be independent sets, if guessed correctly. 

\vspace{10pt}

\noindent{{\bf Lines~\ref{line:while-one-reduce-start}-\ref{line:while-one-reduce-end}}}
For each $i \in \{1, \ldots, |D|\}$, $j \in \{i+1, \ldots, |D|\}$ and $r \in \{1, \ldots, k\}$, 
for every vertex $v \in X_j$ such that $v \in N(\widetilde{X}_{i,j}^r)$, 
delete $r$ from its list, that is update $\lst(v)= \lst(v) \setminus \{r\}$. 
Furthermore, let $r' \in \{1, \ldots, k\} \setminus \{r\}$ such that $rr' \not \in E(H)$.
Then further update $\lst(v)= \lst(v) \setminus \{r'\}$. 
If all the previous guesses were correct, then $\hm(v) \not \in \{r,r'\}$, and therefore this does not change
the solution $C$, that is $\hm: C \to V(H)$ is still a homomorphism that respects the updated list function.
Indeed, because $v \in N(\widetilde{X}_{i,j}^r)$ and $\widetilde{X}_{i,j}^r \subseteq X^i_r$ and, for each $u \in X_i^r$, $\hm(u)=r$. If $\hm(v) \in \{r,r'\}$, 
then $v$ has no neighbors in $\widetilde{X}_{i,j}^r$, which is a contradiction.

\vspace{10pt}

\noindent{{\bf Lines~\ref{line:while-two-reduce-start}-\ref{line:while-two-reduce-end}}}
After the above procedure, 
say there exists an edge $uv \in E(G)$ 
such that $u \in X_i, v \in X_j, r \in \lst(u), r' \in \lst(v)$ and either $r=r'$ or $rr' \not \in E(H)$.
Then remove $r$ from the list of $u$, that is update $\lst(u) = \lst(u) \setminus \{r\}$. The resulting instance is an equivalent instance, 
that is $\hm: C \to V(H)$ is still a homomorphism that respects the updated list function. 
This is
because $u \not \in X_i^r$, 
because if it were then $v$, which is a neighbor of $u$, would be dominated by $\widetilde{X}_{i,j}^r$, and hence $r'$ should have been removed from the $\lst(v)$ by the operation from the previous paragraph.

After doing the clean-ups of the $\lst$ function in the above two paragraphs,
we get to the scenario where for each $X_i,X_j$, $i,j \in \{1, \ldots, |D|\}$, 
if there exists an edge $uv$ where $u \in X_i$ and $v \in X_j$, 
then the lists of $u$ and $v$ are disjoint, and if $r \in \lst(u)$ then for each $r'$ such that $rr'\not \in E(H)$, $r' \not \in \lst(v)$.

\vspace{10pt}

\noindent{{\bf Lines~\ref{line:for-three-reduce}-\ref{line:for-four-end-reduce}}}
Recall that $D \subseteq C$.
For each $d_p \in D$, guess $h(d_p) \in \{1, \ldots, k\}$, 
such that there exists a homomorphism $\hm : C \to V(H)$ that respects $\lst$, 
such that $\hm(d_p) = h(d_p)$.

As a sanity check, discard those guesses where adjacent vertices in $D$ are assigned the same value or values that are not adjacent in $H$.
That is,
for each $d_p \in D$ and  $v \in N_G(d_p)$, update $\lst(v) = \lst(v) \setminus (\{h(d_p)\} \cup \{r': h(d_p)r' \not \in E(H)\})$. 
Also update the $\lst(d_p)=\{h(d_p)\}$.

\vspace{10pt}

\noindent{{\bf Lines~\ref{line:if-size-one-reduce}-\ref{line:if-size-one-reduce-end}, and~\ref{line:output-solution}}}
If $\size(\lst) =1$, then 
for each $r \in \{1, \ldots, k\}$,
let $V_r=\{v \in V(G): \lst(v) = \{r\}\}$.
Let $I_r$ be a maximum weight independent set in $G[V_r]$ computed using the algorithm of Proposition~\ref{prop:is}. Then $\bigcup_{r=1}^k I_r$ is a solution for $(G,\wt,\lst)$.
Indeed, because of all the earlier preprocessing there is indeed a homomorphism from $\bigcup_{r=1}^k I_t$ to $V(H)$ that respects $\lst$. 
Further, for any maximum weight independent set $C$ for which the choices in all previous for-loops were correct, 
$C \cap V_r$ is an independent set and hence $\wt(\bigcup_{r=1}^k I_r) \geq \wt(C)$.
Such a set is then outputted at Line~\ref{line:output-solution}.

\vspace{10pt}

\noindent{{\bf Lines~\ref{line:if-size-two-reduce}-\ref{line:if-size-two-reduce-end}, and~\ref{line:output-family}}}
If $\size(\lst) \geq 2$, then 
for the correct choices in the previous for-loops with respect to the solution $C$,
we get the element $(D, (G[X_1], \wt_{|X_1}, \lst_{|X_1}),  \allowbreak \ldots, \allowbreak (G[X_{|D|}], \allowbreak \wt_{|X_{|D|}},  \allowbreak \lst_{|X_{|D|}}))$ in $\mathcal{F}$.

Since for any $1<i < j \leq |D|$ and any $u \in X_i$ and $v \in X_j$, 
if $r \in \lst(u)$, then $N_H[r] \cap \lst(v) =\emptyset$,
the problem reduces to solving independently on each $G[X_i]$, for $i \in \{1, \ldots, |D|\}$, and taking the union of these solutions with $D$. 
Since $X_i \subseteq N_G(d_i)$, 
for each $v \in X_i$, $\lst(v)  \subseteq \{1, \ldots, k\} \setminus \{h(d_i)\}$ after all the preprocessing.
Therefore, $|\lst(v)|$ has strictly decreased after the above updates.

\end{proof}

\section{Constructing a polynomial-sized family that contains subsets of connected components of a solution}\label{sec:nonconn}

In this section, we prove a lemma that is an analogue of~\cite[Lemma~$1$]{DBLP:conf/soda/AgrawalLLSS24} for \maxpartial. The proof of Lemma~\ref{lem:family} is in fact, an adaptation of the proof of~\cite[Lemma~$1$]{DBLP:conf/soda/AgrawalLLSS24}, but it differs substantially from~\cite[Lemma~$1$]{DBLP:conf/soda/AgrawalLLSS24} in its base case. 
Our base case, that is when we encounter a situation where a connected solution is guaranteed, is far more involved than that of~\cite[Lemma~$1$]{DBLP:conf/soda/AgrawalLLSS24}.
In fact, this is where we resort to Lemma~\ref{lem:conn} recursively to reduce the measure which is the size of the input $\lst$ function.

\begin{lemma}\label{lem:family}
Given a $P_5$-free graph $G'$, a weight function $\wt : V(G) \to \mathbb{Q}_{\geq 0}$ and a list function $\lst: V(G) \to 2^{V(H)}$, 
there is an algorithm that runs 
in $n^{\OO(k'^4)}$ time,
and outputs an $n^{\OO(k'^4)}$-sized collection   
$\mathcal{C} \subseteq 2^{V(G')}$ of \emph{connected} vertex sets of $G'$, where $k'=|V(H)|$,
    such that:  
    \begin{enumerate}
        \item for each $C \in \mathcal{C}$, 
        there exists a homomorphism $\hm : C \to V(H)$ that respects $\lst$,
        and
        \item there exists a set $S \subseteq V(G')$ such that 
        there exists a homomorphism $\hm : S \to V(H)$ that respects $\lst$,
         $\wt(S)$ is maximum
        and $S = \bigcup_{C \in \mathcal{C}'} C$, 
        where $\mathcal{C}' \subseteq \mathcal{C}$
        and for each $C_1, C_2 \in \mathcal{C}'$, $C_1 \cap C_2 =\emptyset$, and $E(C_1,C_2)=\emptyset$.
    \end{enumerate}   
\end{lemma}

The algorithm for Lemma~\ref{lem:family} is described in Algorithm~\ref{alg:key}.
The correctness and the running time bounds of Algorithm~\ref{alg:key} are proved in Lemma~\ref{lem:correctness}. This will prove Lemma~\ref{lem:family}.

\begin{algorithm}[t] 
\caption{Isolating a connected component} 
\label{alg:key} 
\begin{algorithmic}[1] 
    \REQUIRE An undirected graph $G'$,
    a weight function $\wt : V(G') \to \mathbb{Q}_{\geq 0}$ and 
    a list function $\lst: V(G') \to 2^{V(H)}$
    \ENSURE $\mathcal{C} \subseteq 2^{V(G')}$ satisfying the properties of Lemma~\ref{lem:family}

    \STATE Initialize $\mathcal{C}=\{\{v\} : v \in V(G) \}$.\label{line:init}
    \FORALL {$I \subseteq \{1, \ldots, k'\}$}\label{line:minusone}
        \STATE Set $H= H-I$. Let $k= |V(H)|$ and $V(H) =\{1, \ldots,k\}$. \label{line:zero}
        \FORALL{$D \subseteq {V(G') \choose {k \leq i \leq k+1}}$, 
         where $G'[D]$ is connected}\label{line:for-one} 
             \FORALL{$h : D \to \{1, \ldots, k\}$ such that $h^{-1}(r) \neq \emptyset$ for each $r \in \{1, \ldots, k\}$}\label{line:for-two} 
    	           \STATE Initialize $G=G'$.\label{line:init-graph-G}
                    \WHILE{there exists $u \in \bigcap_{r \in \{1, \ldots, k\}} N_G(h^{-1}        (r))$}\label{line:while-one}
                            \STATE Delete $u$ from $G$. That is, $G=G-u$. \label{line:while-one-mid}
                    \ENDWHILE\label{line:while-one-end}
                    \WHILE{there exists a connected component $Z$ of $G-N_G[D]$ such that $Z$ is not a    module in $G$}\label{line:while-two}
                             \STATE Delete $Z$ from $G$. That is, $G=G- Z$.
                    \ENDWHILE\label{line:while-two-end} 
                    \FORALL{$D' \subseteq {V(G) \choose \leq \max\{k+1,3\} }$}\label{line:for-three}
                            \STATE Let $C_{D,D'}^{*} = N[D \cup D']$.\label{line:supersets}
                                    \WHILE{there exists $u \in C_{D,D'}^{*}$ such that $N_{G}(u) \setminus         C_{D,D'}^{*} \neq \emptyset$}\label{line:while-four}
                                          \STATE Delete $u$ from $G$. That is, $G=G-u$ and $C^*_{D,D'} = C^*_{D,D'} \setminus \{u\}$.
                                    \ENDWHILE\label{line:while-four-end}
                            \STATE Let $G^{D,D'} = G[C^*_{D,D'}]$, $\wt^{D,D'}=\wt_{|C^*_{D,D'}}$ and $\lst^{D,D'} = \lst_{|C^*_{D,D'}}$.\label{line:rename}
                            \STATE Run the algorithm of Lemma~\ref{lem:conn} on $(G^{D,D'}, \wt^{D,D'}, \lst^{D,D'})$ for $H$ defined in Line~\ref{line:zero}.\label{line:algo:conn} 
                            \IF{Lemma~\ref{lem:conn} returns a solution, say $C_{D,D'} \subseteq V(G^{D,D'})$}\label{line:size-list-one-start}
                                    \STATE Let $C_{D,D'}^1, \ldots, C_{D,D'}^s$ be the connected components of $C_{D,D'}$.
                                    \STATE Update $\mathcal{C}=\mathcal{C} \cup \{C_{D,D'}^1, \ldots, C_{D,D'}^s\}$.\label{line:update:other}
                            \ELSE \label{line:size-list-two}
                                    \STATE Let $\mathcal{F}$ be the family obtained from Lemma~\ref{lem:conn} on input $(G^{D,D'}, \wt^{D,D'}, \lst^{D,D'})$ for $H$ defined in Line~\ref{line:zero}.\label{line:lem:conn}
                                    \FORALL{$(D'', (G_1, \wt_1, \lst_1), \ldots, (G_p, \wt_p, \lst_p)) \in      \mathcal{F}$}\label{line:for-four}
                                            \STATE For each $q \in \{1, \ldots, p\}$, let $\mathcal{C}_q$ be the family outputted by the recursive call of Algorithm~\ref{alg:key} on input $(G_q, \wt_q,\lst_q)$.\label{line:recursive}
                                            \STATE Let $\mathcal{C}^{D''}=\{\{D'' \cup C_1 \cup \ldots \cup C_p\}: C_q \in \mathcal{C}_q \text{ for each } q \in \{1, \ldots, p\}\}$ be the family of sets obtained by taking the union of one set from each $\mathcal{C}_q$, $q \in \{1, \ldots, p\}$, together with $D''$.\label{line:final}
                                            \STATE Update $\mathcal{C}$ by adding to it each connected component of each element of $\mathcal{C}_{D''}$.\label{line:update}
                                  \ENDFOR
                         \ENDIF
                    \ENDFOR
    \ENDFOR
    \ENDFOR
    \ENDFOR
\end{algorithmic}
\end{algorithm}

\begin{lemma}\label{lem:correctness}
    The family $\mathcal{C}$ outputted by Algorithm~\ref{alg:key} satisfies the properties of Lemma~\ref{lem:family}. Moreover Algorithm~\ref{alg:key} runs in  $n^{\OO(k'^4)}$ time.
\end{lemma}
\begin{proof}
     Let $(G',\wt,\lst)$ be the input instance.
     Note that $|V(H)| =k'$ in the very beginning.
     We first bound the running time of Algorithm~\ref{alg:key}.
     Note that Algorithm~\ref{alg:key} is a recursive algorithm.
     It is called recursively at Line~\ref{line:recursive}.
     Note that excluding the recursive call at Line~\ref{line:recursive}, the call to Lemma~\ref{lem:conn} at Line~\ref{line:algo:conn} and the 5 for-loops, each individual step of Algorithm~\ref{alg:key} runs in $n^{\OO(1)}$.
     Since the recursive calls at Line~\ref{line:recursive} are made on instances received from the family of Lemma~\ref{lem:conn}, that is on instances with strictly smaller list sizes,
      the depth of recursion is at most $|V(H)| =k'$.
     The algorithm of Lemma~\ref{lem:conn} runs in $ n^{\OO(k'^3)}$ time.
    There are 5 for-loops: the one at Line~\ref{line:minusone} runs for $2^{k'}$ times, the one at Line~\ref{line:for-one} runs for $n^{O(k')}$ times, the one at Line~\ref{line:for-two} runs for $k'^{\OO(k')}$ times, the one at Line~\ref{line:for-three} runs for $n^{\OO(k')}$ times and finally the one at Line~\ref{line:for-four} runs for $n^{\OO(k'^3)}$ times because of Lemma~\ref{lem:conn}.
    Since the depth of recursion is bounded by $k'$ and $k' \leq n$ (without loss of generality),
    the overall running time is $n^{\OO(k'^4)}$.

     To bound the size of the family $\mathcal{C}$ that is outputted by Algorithm~\ref{alg:key}, observe that $\mathcal{C}$ is initialized in Line~\ref{line:init} and updated only in Lines~\ref{line:update:other} and/or~\ref{line:update}. 
     At Line~\ref{line:init} the size of $\mathcal{C}$ is $n$ and in each recursive call, Lines~\ref{line:update:other} and~\ref{line:update} are executed at most the number of times all the 5 for-loops mentioned above are executed.
     Also in each execution at most $n$ elements in case of Line~\ref{line:update:other}, and at most $p \cdot n \leq k \cdot n \leq k' \cdot n$ elements in case of Line~\ref{line:final}, are added to $\mathcal{C}$.
     Thus the size of $\mathcal{C}$ is upper bounded by $n^{\OO(k'^4)}$.

    We will now observe that for each set $C \in \mathcal{C}$,
    there exists a homomorphism $\hm: C \to V(H)$ that respects $\lst$.
    Indeed, at Line~\ref{line:init} we add singleton sets, so this holds.
    At Lines~\ref{line:update:other} and~\ref{line:update} this holds because of Lemma~\ref{lem:conn}.

    The remaining proof shows that 
    $\mathcal{C}$ satisfies the second property of Lemma~\ref{lem:family}. 
    Let $S \subseteq V(G')$ such that 
    there exists a homomorphism $\hm: S \to V(H)$ that respects $\lst$,
    $\wt(S)$ is maximum and $S$ has the maximum number of connected components from $\mathcal{C}$. 
    If all connected components of $S$ are in $\mathcal{C}$ then we are done. Otherwise let $C$ be a connected component of $G'[S]$ such that $C \not \in \mathcal{C}$.

    If $C$ has exactly one vertex
    then $C \in \mathcal{C}$ from Line~\ref{line:init}, and hence a contradiction.
    Therefore assume that $|V(C)| \geq 2$. 
    
    Let $I \subseteq \{1, \ldots k'\}$ be such that $i \in I$ if and only if $\hm^{-1}(i) \cap C = \emptyset$.
    Then $\hm_{|C}$, which is $\hm$ restricted to $C$,
    is a homomorphism from $C$ to $V(H)-I$ that respects $\lst$. 
    Henceforth, $H = H-I$. Let $V(H) = \{1, \ldots, k\}$ throughout the following.
    Consider the run of the for-loop at Line~\ref{line:minusone} for this choice of $I$.
    Henceforth, without loss of generality, assume that $\hm^{-1}(i) \cap C \neq \emptyset$ for each $i \in \{1, \ldots,k\}$.

    Since there is a homomorphism from $G[C]$ to $H$, 
    the size of a maximum clique in $G[C]$ is at most  $|V(H)|$ ($=k$).
    Further since $G[C]$ is connected and $P_5$-free,
    from Proposition~\ref{prop:dom} there exists a dominating set $\bar{D}$ of $C$ 
    which either induces a $P_3$ or a clique on at most $k$ vertices.
    We will now show that there exists a dominating set $D \subseteq C$ of $C$ of size $k$ or $k+1$,
    such that $G[D]$ is connected and 
    for each $r \in \{1, \ldots, k\}$,
    $\hm^{-1}(r)\cap C \cap D \neq \emptyset$.
    
    Suppose that $\bar{D}$ is a clique. 
    In this case for each $r \in \{1, \ldots, k\}$
    there exists at most one vertex $d_r \in \bar{D}$, 
    such that $d_r \in \hm^{-1}(r) \cap C$.
    For each $r' \in \{1 \ldots, k\}$ such that 
    $\bar{D} \cap C \cap \hm^{-1}(r') = \emptyset$,
    fix $d_{r'} \in \hm^{-1}(r') \cap C$
    such that $d_{r'} \in N(\bar{D})$.
    Such a vertex $d_{r'}$ exists because
     $\bar{D}$ is a dominating set of $C$ and $\hm^{-1}({r'}) \cap  C \neq \emptyset$.
    Finally, set $D := \bar{D} \cup \{d_{r'} : \bar{D} \cap C \cap \hm^{-1}(r') = \emptyset \}$.
    Then $D$ is still a dominating set of $C$ of size $k$ such that for each $r \in \{1, \ldots, k\}$,
    there exists $d_r \in D$ such that $d_r \in \hm^{-1}(r) \cap C$. 
    Consider the execution of  
     the for-loop at Line~\ref{line:for-one} for this $D$ when $\bar{D}$ is a clique.
    
    Otherwise if $\bar{D}$ is an induced $P_3$, 
    repeat the above process to find a superset of $\bar{D}$ that intersects each $\hm^{-1}(r) \cap C$, for $r \in \{1, \ldots, k\}$. Since originally at most two vertices of the induced $P_3$ ($\bar{D}$) can intersect some $\hm^{-1}(r') \cap C$,
    we can find a set $D$ superset of $\bar{D}$ which is connected and intersects all $\hm^{-1}(r) \cap C$ for each $r \in \{1, \ldots, k\}$,
    of size at most $k+1$ (and at least $k$). 
    In the case that $\bar{D}$ is an induced $P_3$,
    consider the execution of the for-loop at Line~\ref{line:for-one} for this $D$.

    Let $h : D \to \{1, \ldots, k\}$ such that for each $d \in  D$, $h(d) = \hm(d)$.
    Consider the execution of the for-loop for this function $h$. Because of the choice of $D$,
    for each $r \in \{1, \ldots, k\}$,
    $h^{-1}(r) \neq \emptyset$.

    \vspace{10pt}

     \noindent{\bf [Lines~\ref{line:while-one}-\ref{line:while-one-end}]} Note that any vertex $u \in \bigcap_{r \in \{1, \ldots, k\}} N(h^{-1}(r))$ is not in $C$ because,
     for each $r \in \{1, \ldots, k\}$, $h^{-1}(r) \neq \emptyset$. 
     Furthermore,
    such a vertex $u$ does not belong to the solution $S$ 
    because $u \in N(C)$ (since $D \subseteq C$) and $C$ is a connected component of $G[S]$.

    \vspace{10pt}

    \noindent{\bf [Lines~\ref{line:while-two}-\ref{line:while-two-end}]} We will now show that any connected component of $G-N[D]$ which is not a module, is a subset of $N(C)$. This will imply that $S$ is also a solution in $G- Z$, where $Z$ is a connected component of $G-N[D]$ which is not a module. 
    Let $X=N(C) \setminus N(D)$ and let $Y= V(G) \setminus N[C]$. 
    
    \begin{claim}\label{claim:no-across-edges}
    $E(X,Y) = \emptyset$.
    \end{claim}
    \begin{claimproof}
    Suppose for the sake of contradiction that there exists $y \in Y$ and $x \in X$ such that $yx \in E(G)$. 
    Since $x \in X$ and $X = N(C)\setminus N(D)$, 
    there exists $c \in C$ such that $xc \in E(G)$. Also $c \not \in D$, as otherwise $x \in N(D)$. 
    Say $c \in \hm^{-1}(i) \setminus D$, for some $i \in \{1, \ldots, k\}$. 
    Since $D$ is a dominating set of $C$, 
    there exists $d \in \hm^{-1}(j)$, where $j \neq i$, 
    such that $cd \in E(G)$.
    Let $d' \in \hm^{-1}(i) \cap D$. Such a vertex $d'$ exists because $D$ intersects every $\hm^{-1}(r)$ for $r \in \{1, \ldots, k\}$.
    Also because $D$ is connected, there exists a shortest path $P'$ from $d$ to $d'$ in $G[D]$. Let $d^*$ be the last vertex on this path such that $cd^* \in E(G)$. 
    Note that $d^* \neq d'$ because $c,d \in \hm^{-1}(i)$ and $G[\hm^{-1}(i)]$ is independent. Let $P$ be the subpath of $P'$ from $d^*$ to $d'$ in $G[D]$,

    Consider the path $P^*=(y,x,c,P')$. We claim that $P^*$ is an induced $P_t$ where $t \geq 5$. This will be a contradiction.
    First observe that all the vertices of $P^*$ except $y,x$ are in $C$. 
    Also $E(Y,C) =\emptyset$ because $Y \cap N(C) = \emptyset$ by the definition of $Y$. 
    In particular, $yc, y\tilde{d} \not \in E(G)$ for each $\tilde{d} \in P$. 
    Since $x \in X$ and $X \cap N(D) = \emptyset$, $x\tilde{d} \not \in E(G)$ for each $\tilde{d} \in P$. 
    Since $c,d' \in \hm^{-1}(i)$, $cd' \not \in E(G)$. Finally for each vertex $\tilde{d} \in P \setminus {d^*}$, $c\tilde{d} \not \in E(G)$, by definition of $P$.
    \end{claimproof}

    From Claim~\ref{claim:no-across-edges}, for any connected component $Z$ of $G-N[D]$, either $Z \subseteq X$, or $Z \subseteq Y$.

    \begin{claim}\label{claim:module}
    Let $Y'$ be a connected component of $G[Y]$. Then $V(Y')$ is a module in $G$.
    \end{claim}
    \begin{claimproof}
    First note that, from Claim~\ref{claim:no-across-edges}, $N_G(Y) \subseteq N(D)\setminus C$. 
   For the sake of contradiction, say $V(Y')$ is not a module, and thus there exists $y_1,y_2 \in V(Y')$ and $u \in N(D) \setminus C$ such that $y_1u \not \in E(G)$ and $y_2u \in E(G)$. Also choose such a pair $y_1,y_2 \in Y$ such that their distance in $Y$ is smallest.
   Then $y_1y_2 \in E(G)$.
   Since $u \in N(D) \setminus C$ and $D \subseteq C$, 
   there exists $d \in D$ such that $ud \in E(G)$. 
   Without loss of generality let $d \in \hm^{-1}(i)$. 
    Since $u \not \in \bigcap_{r \in \{1, \ldots, k\}} N(h^{-1}(r))$ (as otherwise $u$ would have been deleted at Line~\ref{line:while-one-mid}), 
    there exists $d' \in \hm^{-1}(j) \cap D$, $j \neq i$, such that
$ud' \not \in E(G)$. Since $G[D]$ is connected and $d,d' \in D$,
there exists a path from $d$ to $d'$. Consider a subpath $P$ on this path from $d^*$ to $d'$ where $d^*$ is the last vertex on this path which has an edge to $u$.
 Then $P^*=(y_1, y_2, u, P)$ is an induced $P_t$ in $G$ for $t \geq 5$, which is a contradiction.
    \end{claimproof}

From Claim~\ref{claim:module} if a connected component $Z$ of $G-N[D]$ is not a module, then $Z \subseteq X$. Since $X \subseteq N(C)$, $S$ is also an induced subgraph of $G$ after the execution of Line~\ref{line:while-two}-\ref{line:while-two-end}. 

\vspace{10pt}

\noindent{{\bf [Line~\ref{line:for-three}]}} Let $R =N(C) \setminus N(D)$, i.e., $R$ is the set of remaining vertices of $N(C)$ that are not in $N(D)$ (note that $R$ are precisely the vertices of $X$ from the previous discussion that are not deleted at Line~\ref{line:while-two}-\ref{line:while-two-end}). We now show that there exists a small dominating set of $G[R \cup C]$.

\begin{claim}\label{claim:ds-c-and-x}
$G[R \cup C]$ has a dominating set $\widetilde{D}$ of size at most $k+1+ \max\{k+1,3\}$ such that $D\subseteq \widetilde{D}$ and $|\widetilde{D} \setminus D| \leq \max\{k+1,3\}$. 
\end{claim}
\begin{claimproof}
    Since $G[R \cup C]$ is a connected and $P_5$-free graph, 
    from Proposition~\ref{prop:dom}, there exists a dominating set of $G[R \cup C]$ which is either a clique or an induced $P_3$. If $D$ itself is a dominating set of $G[R \cup C]$, then the claim trivially follows. 
    
    Otherwise, we consider a dominating clique or a dominating induced $P_3$, $\bar{D}$ of $G[R \cup C]$ of minimum possible size. 
    If $\bar{D}$ induces a $P_3$ then, $\widetilde{D} = \bar{D}\cup D$ satisfies the requirement of the claim. 
    
    Now we consider the case when $\bar{D}$ is a clique. 
    Using $\bar{D}$ we will construct a dominating set $\widetilde{D}$ of $G[R \cup C]$ with at most $2(k+1)$ vertices, which contains the vertices from $D$.
    Intuitively speaking, apart from $D$ (whose size is at most $k+1$) we will add vertices of $\bar{D}\cap C$ and at most one more vertex, to construct $\widetilde{D}$. 
    We remark that since $\bar{D}$ is a clique and the size of a maximum clique in $G[C]$ is a at most $k$,
     $|\bar{D}\cap C| \leq k$. 
 
    Let $R_1, \ldots, R_p$ be the connected components of $G[R]$. 
    Since $\bar{D}$ is a clique, $\bar{D}$ intersects at most one $R_i$, that is, there exists an $i \in \{1,\ldots,p\}$, such that $\bar{D} \cap V(R_j) =\emptyset$, for all $j \in \{1,\ldots,p\}\setminus \{i\}$.
    Therefore the vertices of $\bigcup_{j \in \{1,\ldots,p\} \setminus \{i\}}V(R_i)$ are dominated by the vertices of $\bar{D} \cap C$. If $\bar{D}\cap V(R_i) =\emptyset$, then notice that $\bar{D}\subseteq C$, where $|\bar{D}|\leq k$,
    and thus $\widetilde{D} = \bar{D}\cup D$ satisfies the requirement of the claim. 
    Now suppose that $\bar{D}\cap V(R_i) \neq \emptyset$, 
    and consider a vertex $x \in \bar{D} \cap V(R_i)$. 
    Recall that $x\in N(C)$ and $R_i$ is a module in $G$. 
    Thus, there exists a vertex $v' \in C\cap N(x)$, and moreover, we have $V(R_i)\subseteq N(v')$. 
    Note that $D$ dominates each vertex in $C$, 
    $v'$ dominates each vertex in $R_i$, 
    and $\bar{D}\cap C$ dominates each vertex in $\bigcup_{j \in \{1,\ldots,p\} \setminus \{i\}} V(R_i)$. Thus, $\widetilde{D} = D \cup \{v'\} \cup (\bar{D} \cap C)$ dominates each vertex in $G[R \cup C]$. Further $|D' \setminus D| \leq k+1$. This concludes the proof.  
\end{claimproof}

Fix $\widetilde{D}$ which is dominating set of $G[R \cup C]$ obtained from Claim~\ref{claim:ds-c-and-x}.
Fix $D' := \widetilde{D} \setminus D$.
Now consider the execution of the for-loop at Line~\ref{line:for-three} for this $D'$.

\vspace{10pt}

\noindent{{\bf [Line~\ref{line:supersets}]}} We now show that the set $C^*_{D,D'}$ 
at Line~\ref{line:supersets} is equal to $ N[C]$ in (the reduced/most updated) $G$. To obtain the above, it is enough to show that $N[\widetilde{D}] = N[C]$. To this end, we first obtain that $N[C] \subseteq N[\widetilde{D}]$. Recall that, after removing the vertices from $X$ at Line 8, $N[C] = R \cup C \cup N(D)$. As $\widetilde{D} = D \cup D'$ is a dominating set for $G[R\cup C]$, we have $R \cup C \subseteq N[\widetilde{D}]$. Moreover, as $D\subseteq \widetilde{D}$, we have $N(D) \subseteq N[\widetilde{D}]$. Thus we can conclude that $N[C] \subseteq N[\widetilde{D}]$. We will next argue that $N[\widetilde{D}] \subseteq N[C]$. Recall that from Claim~\ref{claim:no-across-edges}, $E(R,Y) =\emptyset$. Thus, for any vertex $v\in D' \setminus C$, $N(v) \subseteq N[C]$. In the above, when $v\in D' \setminus C$, without loss of generality we can suppose that $v\in R \subseteq N(C)$, as $\widetilde{D} = D \cup D'$ is a dominating set for $G[R\cup C]$. Thus, for each $v\in D' \setminus C$, $N[v] \subseteq N[C]$. Also, $D\subseteq C$, and thus, $N[D]\subseteq N[C]$. Hence it follows that $N[\widetilde{D}] \subseteq N[C]$. Thus we obtain our claim that, $N[\widetilde{D}] = N[C]$.

\vspace{10pt}

\noindent{{\bf [Lines~\ref{line:while-four}-\ref{line:while-four-end}]}}
Since $C^*_{D,D'} = N[C]$, if there exists $u \in C^*_{D,D'}$ such that $u$ has a neighbor outside $C^*_{D,D'}$, then $u \in N(C)$ and hence $u \not \in S$. Thus $S$ is a solution even in the graph obtained by deleting such vertices. Notice that after the execution of these steps, $N[C]$ is a connected component of $G$, as removing a vertex from $N(C)$ cannot disconnect the graph $N[C]$. 

\vspace{10pt}

\noindent{{\bf [Lines~\ref{line:rename}-\ref{line:algo:conn}]}}
Henceforth, for simplicity of notation let $G^{D,D'}= G[C^*_{D,D'}]$, $\wt^{D,D'}=\wt_{|C^*_{D,D'}}$ and $\lst^{D,D'} = \lst_{|C^*_{D,D'}}$.
Run the algorithm of Lemma~\ref{lem:conn} on the instance $(G^{D,D'}, \wt^{D,D'}, \lst^{D,D'})$.

\vspace{10pt}

\noindent{{\bf [Lines~\ref{line:size-list-one-start}-\ref{line:update:other}]}}
If 
Lemma~\ref{lem:conn} outputs a maximum weight solution $C_{D,D'}$ of $(G^{D,D'}, \wt^{D,D'}, \lst^{D,D'})$, then
let $S' = (S \setminus C) \cup C_{D,D'}$.
Then $\wt(S') \geq \wt(S)$. Also $G[S']$ is a solution 
because there exists a homomorphism from $G^{D,D'}$ to $H$ that respects $\lst^{D,D'}$
and $E(S\setminus C, C_{D,D'}) =\emptyset$ (because $N[C]$ is a connected component of the reduced graph $G$). Additionally $S'$ has more connected components in $\mathcal{C}$ compared to $S$, which contradicts the choice of $S$.

\vspace{10pt}

\noindent{{\bf [Lines~\ref{line:size-list-two}}-\ref{line:update}]} 
If the algorithm of Lemma~\ref{lem:conn}, revoked on the instance $(G^{D,D'},\wt^{D,D'}, \lst^{D,D'})$, outputs the family $\mathcal{F}$, then 
Lemma~\ref{lem:conn} guarantees that there exists a maximum weight solution $C'$ in $G^{D,D'}$ such that for an element $(D'', (G_1,\wt_1,\lst_1) , \ldots, (G_p,\wt_p,\lst_p))$, where $p \leq k$, 
$C'$ is the union of $D''$ and arbitrary solutions for $(G_1,\wt_1,\lst_1), \ldots, (G_p,\wt_p,\lst_p)$.
Thus $C'$ belongs to the family $\mathcal{C}_{D''}$ at Line~\ref{line:final}, and the connected components of $C'$ belong to $\mathcal{C}$ (from Line~\ref{line:update}).
Let $S'= (S\setminus C) \cup C'$. Then $\wt(S') \geq \wt(S)$. Also $G[S']$ is a solution because there exists a homomorphism from $G^{D,D'}$ to $H$ that respects $\lst^{D,D'}$
and $E(S\setminus C, C') =\emptyset$ (because $N[C]$ is a connected component of the reduced graph $G$). Additionally $S'$ has more connected components in $\mathcal{C}$ compared to $S$, which contradicts the choice of $S$.
\end{proof}

\section{Proof of Theorem~\ref{thm:main} by reduction to the \textsc{Maximum Weight Independent Set} problem on the blob graph}\label{sec:proof-of-theorem}

Run the algorithm of Lemma~\ref{lem:family} on the instance $(G,\wt,\lst)$. Let $\mathcal{C}$ be the family returned. Recall that all vertex sets in $\mathcal{C}$ are connected.
We say that two sets $C_1, C_2 \in \mathcal{C}$ {\em touch} each other, if either $C_1 \cap C_2 \neq \emptyset$ or $E(C_1,C_2) \neq \emptyset$. 
Let $G^{\blob}_{{\mathcal{C}}}$ be an auxiliary graph whose vertex set corresponds to sets in ${\mathcal{C}}$ and there is an edge between two vertices of $G^{\blob}_{{\mathcal{C}}}$ if and only if the corresponding sets touch other. Formally, $V(G^{\blob}_{{\mathcal{C}}}) =\{v_C : C \in {\mathcal{C}}\}$ and for any $v_C,v_{C'} \in V(G^{\blob}_{{\mathcal{C}}})$, there exists $(v_C,v_{C'}) \in E(G^{\blob}_{{\mathcal{C}}})$ if and only if $C$ and $C'$ touch each other. Let $\wt^{\blob}: V(G^{\blob}_{{\mathcal{C}}}) \to \mathbb{Q}_{\geq 0}$ be defined as follows: $\wt^{\blob}(v_C) = \sum_{u \in C} \wt(u)$.

\begin{proposition}[Theorem~$4.1$,~\cite{Gartlandetal2021}]\label{prop:blob}
    If $G$ is $P_5$-free, then $G^{\blob}_{{\mathcal{C}}}$ is also $P_5$-free.
\end{proposition}

The following lemma reduces the task of finding a maximum weight solution in $G$ to that of finding maximum weight independent set in $G^{\blob}_{\mathcal{C}}$.

\begin{lemma}\label{lem:oct-equiv-packing}
$(G,\wt, \lst)$ has a solution of weight (defined by $\wt$) $W$ if and only if there exists $\mathcal{C}' \subseteq {\mathcal{C}}$ such that no two sets in $\mathcal{C}'$ touch each other and $\sum_{C \in \mathcal{C}'} \wt(C) = W$.
\end{lemma}
\begin{proof}
    From Lemma~\ref{lem:family}, there exists $S \subseteq V(G)$ such that $S$ is a solution, the weight of $S$ (with respect to $\wt$) is maximum, 
    and all the connected components of $G[S]$ are contained in ${\mathcal{C}}$. 
    Therefore $S$ corresponds to a pairwise non-touching sub-collection of ${\mathcal{C}}$ of total weight equal to $\wt(S)$. 
    For the other direction, since every set $C$ in ${\mathcal{C}}$ is connected and there is a homomorphism from $G[C]$ to $V(H)$ respecting $\lst$, 
    we conclude that the union of the sets in any sub-collection of ${\mathcal{C}}$, that contains no two sets that touch each other, has a homomorphism to $V(H)$ respecting $\lst$.
\end{proof}

The following lemma is an immediate consequence of Lemma~\ref{lem:oct-equiv-packing}.
\begin{lemma}\label{lem:oct-equiv-is}
$(G,\wt,\lst)$ has a solution of weight (with respect to the weight function $\wt$) $W$ if and only if $G^{\blob}_{{\mathcal{C}}}$ has an independent set of weight $W$, with respect to the weight function $\wt^{\blob}$.
\end{lemma}

From Proposition~\ref{prop:blob},  if $G$ is $P_5$-free, then $G^{\blob}_{{\mathcal{C}}}$ is also $P_5$-free. 
So by Lemma~\ref{lem:oct-equiv-is}, the problem actually reduces to finding a maximum weight independent set in a $P_5$-free graph, which can be done by~\cite{DBLP:conf/soda/LokshantovVV14}.
Now we are ready to prove our main result, that is, Theorem~\ref{thm:main}.

\begin{proof}[Proof of Theorem~\ref{thm:main}]
Let $(G,\wt,\lst)$ be an instance of \maxpartial. 
Let $\mathcal{C}$ be the family of connected sets of $G$ returned by Lemma~\ref{lem:family} on input $(G, \wt,\lst)$.
Construct an instance $(G^{\blob}_{{\mathcal{C}}},\wt^{\blob})$ as described earlier.
Note that the number of vertices of $G^{\blob}_{{\mathcal{C}}}$ is $|{\mathcal{C}}|$. 
Since $|\mathcal{C}| \leq k^k \cdot n^{\OO(k)}$ by Lemma~\ref{lem:family}, 
the number of vertices of $G^{\blob}_{{\mathcal{C}}}$ is at most $ k^k \cdot n^{\OO(k)}$. 
Also the construction of this graph takes time polynomial in $k^k \cdot n^{\OO(k)}$.

Thus, by Lemma~\ref{lem:oct-equiv-is}, the problem reduces to finding a maximum weight independent set in a $P_5$-free graph $G^{\blob}_{{\mathcal{C}}}$. A maximum weight independent set on $P_5$-free graphs can be found in $\OO(n^{12} m)$-time using Proposition~\ref{prop:is}. This concludes the proof of Theorem~\ref{thm:main}.
\end{proof}

\bibliographystyle{plain}
\bibliography{references}

\end{document}